\documentclass[letterpaper, 10 pt, conference]{ieeeconf}  

\IEEEoverridecommandlockouts                              
\usepackage{amsmath} 
\usepackage{amssymb}  

\usepackage{hyperref}
\usepackage{graphicx}
\usepackage{algorithm}
\usepackage{algpseudocode}
\usepackage{subfig}

\title{\LARGE \bf
\NAME~for\\ Uncertain Nonlinear Control Systems
}

\author{Lukas Brunke, Siqi Zhou, and Angela P. Schoellig
\thanks{The authors are with the Dynamic Systems Lab
	(http://www.dynsyslab.org), Institute for Aerospace Studies,
	University of Toronto, Canada. The authors are also affiliated with
	the University of Toronto Robotics Institute and the Vector Institute for Artificial Intelligence, Toronto. Angela P. Schoellig is also with the Department of Electrical and Computer
Engineering, Technical University of Munich, Germany. Emails:
	\{lukas.brunke, siqi.zhou, angela.schoellig\}@robotics.utias.utoronto.ca}%
}

\newcommand{\acronym}{RPOF-SF}

\newcommand{\NAME}{Robust Predictive Output-Feedback Safety Filter}
\newcommand{\NAMECAPS}{ROBUST PREDICTIVE OUTPUT-FEEDBACK SAFETY FILTER}

\newcommand{\R}{\mathbb{R}}
\newcommand{\set}[1]{\mathbb{#1}}

\newtheorem{theorem}{Theorem}

\newtheorem{definition}{Definition}

\newtheorem{assumption}{Assumption}

\usepackage{cite}

\begin{document}

\maketitle
\thispagestyle{empty}
\pagestyle{empty}

\begin{abstract}
In real-world applications, we often require reliable decision making under dynamics uncertainties using noisy high-dimensional sensory data. Recently, we have seen an increasing number of learning-based control algorithms developed to address the challenge of decision making under dynamics uncertainties. These algorithms  often make assumptions about the underlying unknown dynamics and, as a result,  can provide safety guarantees. This is more challenging for other widely used learning-based decision making algorithms such as reinforcement learning. Furthermore, the  majority of existing approaches assume access to  state measurements, which can be restrictive in practice. In this paper, inspired by the literature on safety filters and robust output-feedback control, we present a robust predictive output-feedback safety filter (RPOF-SF) framework that provides safety certification to an arbitrary controller applied to an uncertain nonlinear control system. The proposed RPOF-SF combines a robustly stable observer that estimates the system state from noisy measurement data and a predictive safety filter that  renders an arbitrary controller safe by (possibly) minimally modifying the controller input to guarantee safety. We show in theory that the proposed RPOF-SF guarantees constraint satisfaction despite disturbances applied to the system. We demonstrate the efficacy of the proposed RPOF-SF~algorithm using an uncertain mass-spring-damper system. 
\end{abstract}

\section{INTRODUCTION}
In many practical settings including autonomous driving and other robotics applications, only noisy and high-dimensional measurements are available  and controllers must be designed that provide desired safety guarantees despite not having perfect state information. However, many  control techniques assume access to the full state and an exact description of the system dynamics and observation model, which may not be available. 

Recently, learning-based controllers have gained interest in such settings as they can synthesize control policies from high-dimensional measurements. Learning-based controllers such as reinforcement learning~(RL) agents can improve performance by leveraging past experiences from the interaction with an environment, e.g., by operating in the real world.
These experiences generally also include failures, which could result in damage  to the robot or 
of its surrounding. 

As learning-based controllers often do not account for safety constraints, add-on safety filters have been proposed to decouple performance~(learning-based controller) and safety~(safety filter). These safety filters certify the safety of control inputs from learning-based controllers and modify them if the safety filter determines the original control input would lead to safety violations. 
Current safety filter methods typically assume noisy state feedback. However, in practice, we must often rely on partial and noisy output measurements instead. 
If a safety filter does not account for measurement noise and/or state estimation errors present in the system, then the safety filter can generally not guarantee safety~\cite{DSL2021}. 
\begin{figure}[t]
    \centering
    \includegraphics[width=\columnwidth]{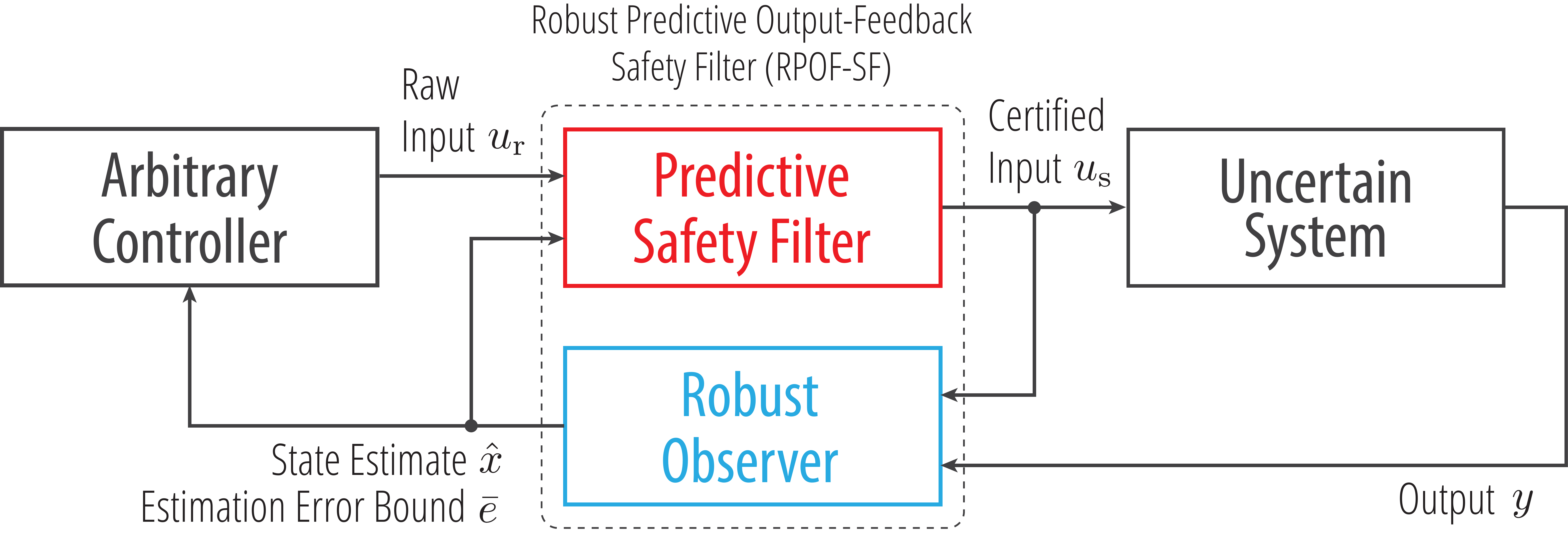}
    \vspace{-1em}
    \caption{A block diagram of the proposed robust predictive output-feedback safety filter (\acronym). The proposed \acronym~consists of a robust observer (blue) and a predictive safety filter module (red). The robust observer provides an estimate of the system state from noisy output measurements, while the predictive safety filter modifies the input sent from an arbitrary controller with the goal to guarantee constraint satisfaction (i.e., to make it safe).}
    \label{fig:blockdiagram}
\end{figure}

In this paper, building on the idea of robust output feedback model predictive control~(MPC)~\cite{Kohler2021-ROFMPC} and model predictive safety certification~(MPSC)~\cite{Wabersich2021}, we derive a robust predictive output-feedback safety filter~(\acronym) that includes a predictive safety filter and a robustly stable observer~(see~\autoref{fig:blockdiagram}).
The proposed approach allows an uncertain system to operate safely under arbitrary control policies.  
We validate the proposed \acronym~on a numerical example\footnote{The code is available at: https://github.com/utiasDSL/dsl\_\_rpof\_\_sf}.
 
\section{Related Literature}
 We summarize related literature on safety filters and robust output-feedback control in this section.
\subsection{Safety Filters}
Learning-based controllers such as reinforcement learning have been proposed to address the challenge of decision making under uncertainties. As compared to traditional model-based control techniques, learning-based approaches can be applied to a wider class of systems subject to model uncertainties~\cite{DSL2021}. Despite their flexibility, learning-based controllers often lack formal safety guarantees. In recent years, there have been several efforts from the control community targeted to address the problem of providing desired theoretical guarantees for learning-based controllers that are \textit{not} initially designed to be safe. One stream of work in this area are safety filters. 
A safety filter minimally modifies any unsafe control inputs from the learning-based controller such that the system's state stays inside a safe set~\cite{DSL2021}, which is either explicitly or implicitly defined. 

Common techniques for explicitly defining safe sets include control barrier functions~(CBF)~\cite{ames2019a, taylor2020b, l4dc22} and Hamilton-Jacobi (HJ) reachability~\cite{mitchell2005time, Fisac2019a}, which are generally defined for continuous-time systems. Intuitively, the CBF framework provides a scalar condition for certifying control inputs to guarantee state constraint  satisfaction, while the HJ reachability framework provides a means to compute a robust positive control invariant safe set, which is a set that is contained in a given state constraint set such that if the uncertain system starts inside the set, there exists a control law to keep the system inside the set despite disturbances. The explicit safe set characterizations from the CBF framework and the HJ reachability analysis provide the basis for designing safety filters to render a learning-based control system safe. 
Recently, CBF and HJ reachability are combined with learning to reduce the conservatism in the safety filter design for uncertain systems~\cite{taylor2020b,l4dc22,Fisac2019a, 6225136}.  


An alternative to CBF and HJ reachability frameworks for safety certification is through the use of predictive filters, which do not assume a pre-computed safe set but determine the safe set implicitly via a MPC framework. A common safety filter that implements this idea is MPSC, which solves a finite-horizon constrained optimization problem with a discrete-time predictive model to prevent a learning-based controller from violating constraints~\cite{Wabersich2021}. These predictive filters typically require state-feedback and do not handle partial state measurements. 
%
To the best of our knowledge, a predictive safety filter for output-feedback systems has not been proposed.

\subsection{Robust Output-Feedback Control} 
Due to their wide applicability and popularity, we focus on MPC techniques in this section. At every time step, MPC solves a finite-horizon optimal control problem subject to constraints, applies the first optimal control input, and replans at the next time step~\cite{Rawlings2017}. 
As MPC relies on its predictive model for optimal performance, disturbances acting on the system can lead to loss of feasibility of the optimization or constraint violation at the next time step. 
Under the assumption of full-state measurements, recent nonlinear robust tube-based MPC schemes guarantee constraint satisfaction and feasibility for uncertain systems~\cite{Kohler-robust-mpc, SinghLandryEtAl2019, Bayer2013}. These schemes rely on a pre-stabilizing controller to stay close to a nominal trajectory and constraint tightening, such that no disturbance can lead to constraint violations. 
For linear systems with partial and noisy state measurements, there exist robust output-feedback MPC methods~\cite{Mayne2006, Kogel2017, Lorenzetti2020, cdc21}. These approaches combine a robust MPC with a Luenberger observer and apply additional constraint tightening based on bounds on the estimation error.
Recent work in~\cite{Kohler2021-ROFMPC} extends the nonlinear robust tube-based MPC scheme to the output-feedback setting. This approach designs a robustly stable observer that predicts verifiable error bounds on the estimation error. Additionally accounting for these bounds in the pre-stabilizing controller and constraint tightening yields constraint satisfaction for all future time steps. 

There also exist min-max robust MPC formulations that simultaneously optimize for the optimal state estimate and control input sequence~\cite{Copp2017, Lofberg2002}. While~\cite{Lofberg2002} is only applicable to linear systems,~\cite{Copp2017} handles nonlinear systems. However, these techniques typically have a high computational demand. This puts min-max methods at a disadvantage compared to tube-based methods for systems with complex dynamics requiring high control rates.  


In this paper, we combine  model predictive safety certification~\cite{Wabersich2021} with a recent framework for robust output-feedback MPC~\cite{Kohler2021-ROFMPC} to guarantee safe closed-loop operation for arbitrary control policies, which includes potentially unsafe control policies of a reinforcement learning agent.  



\subsubsection*{\textbf{Notation}} 
The non-negative real numbers are $\R_{\geq 0}$. 
The set of integers in the interval $\left[a, b\right] \subset \R$ is~$\set{I}_{\left[a, b\right]}$, and the set of integers in the interval $\left[a, \infty \right) \subset \R$ is~$\set{I}_{\geq a}$. 
The class $\mathcal{K}$ denotes continuous functions $\gamma : \left[0, a\right) \to \R_{\geq 0}$ with $a > 0$,  $\gamma(0) = 0$, and $\gamma$ strictly monotonically increasing. The class $\mathcal{K}_\infty$ denotes functions $\gamma \in \mathcal{K}$ with $a = \infty$ that satisfy $\lim_{r \to \infty} \gamma(r) = \infty$.  

\section{PROBLEM STATEMENT}
This paper is concerned with safely controlling an uncertain nonlinear constrained system with multiple inputs and multiple outputs, where not all states are measured and there is measurement noise. 
We consider the following uncertain nonlinear discrete-time system 
\begin{align}
\label{eq:system}
    \begin{split}
        x_{k + 1} &= f_w(x_k, u_k, w_k) = f(x_k, u_k) + E w_k \\
        y_k &= h_w(x_k, u_k, w_k) = h(x_k, u_k) + F w_k \,,
    \end{split}
\end{align}
where $k \in \set{I}_{\geq 0}$ is the the discrete time step, $x_k \in \set{X} \subseteq \R^{n_x}$ is the state, $u_k \in \set{U} \subseteq \R^m$ is the control input, $w_k \in \set{W} \subset \R^q$ is the disturbance with $\set{W}$ being a bounded disturbance set containing zero, and $y_k \in \set{Y} \subseteq \R^{n_y}$ is the noisy output. As $\set{W}$ is bounded, there exists a scalar bound on the disturbance $\bar{w} \geq 0$ such that $\lVert w_k \rVert \leq \bar{w}$ for all $k \in \set{I}_{\geq 0}$. The state dynamics $f_w : \set{X} \times \set{U} \times \set{W} \to \R^{n_x}$ and output equation $h_w : \set{X} \times \set{U} \times \set{W} \to \R^{n_y}$ are continuous functions. The nominal dynamics and output equations are given by $f(x, u) = f_w(x, u, 0)$ and $h(x, u) = h_w(x, u, 0)$, respectively. The matrices $E \in \R^{n_x \times q}$ and $F \in \R^{n_y \times q}$ map disturbances to states and outputs, respectively.

The goal of this work is to design a robust predictive output-feedback safety filter that guarantees safe operation of the system under arbitrary and potentially unsafe control policies. Safe operation of the system is defined by staying inside user-defined state and input constraints $(x_k, u_k) \in \set{Z}$ for all time steps $k~\in~\set{I}_{\geq 0}$. In this paper, we leverage techniques from a recent robust output-feedback MPC framework~\cite{Kohler2021-ROFMPC} to develop a model predictive safety certification scheme for uncertain nonlinear systems without full-state feedback and noisy measurements.  

\section{BACKGROUND}
In this section, we introduce the necessary background for our proposed robust predictive output-feedback safety filter, which we will refer to as the \acronym. 
First, we introduce a recent framework for robust output-feedback MPC~\cite{Kohler2021-ROFMPC}. This framework leverages incremental Lyapunov functions, which provide scalar propagation laws for over-approximations of the estimation and prediction errors. These propagation laws allow the design of robustly stable observers and can be efficiently evaluated in the MPC optimization. Second, we present a predictive safety filter inspired by MPC as in~\cite{Wabersich2021}. We leverage these ideas in our \acronym~to guarantee safe operation of dynamics systems with constraint satisfaction under bounded uncertainties in the dynamics and output. 

\subsection{Robust Output-Feedback MPC}
We first introduce recent results on robust output-feedback MPC for nonlinear uncertain systems using online estimation error bounds~\cite{Kohler2021-ROFMPC}.

\subsubsection*{\textbf{Incremental Lyapunov Functions}}
Incremental Lyapunov functions provide a way of analyzing the stability between two trajectories from the same system with the same control input but different initial conditions. Intuitively, incremental stability between two trajectories means that both trajectories converge to the same trajectory. 
The notion of incremental stability can be extended to \textit{(i)}~incremental input/output-to-state stability~(i-IOSS) by considering systems with inputs and outputs or \textit{(ii)} incremental input-to-state stabilizability by considering a system with an additional feedback policy~\cite{Kohler2021-ROFMPC}.
We begin with the definitions of detectability and stabilizability in terms of incremental Lyapunov functions.

Detectability for nonlinear systems can be described by an i-IOSS Lyapunov function~\cite{allan2021}. 
As we will show below, the i-IOSS Lyapunov function provides us with the tool to synthesize robustly stable observers with guaranteed estimation error bounds.
\definition[i-IOSS Lyapunov function~\cite{allan2021}]{
\label{def:i-ioss-lyap}
A function $V_\mathrm{d} : \set{X} \times \set{X} \to \R_{\geq 0}$~(with subscript~d for detectability) is called an (exponential-decrease) i-IOSS Lyapunov function if there exist lower and upper bounds on $V_{\mathrm{d}}$ with $\alpha_{\mathrm{d, l}}, \alpha_{\mathrm{d, u}} \in \mathcal{K}_{\infty},$ bounding functions $ \sigma_{
\mathrm{d}, w}, \sigma_{\mathrm{d}, y} \in \mathcal{K}$, and a decay rate $\rho_{\mathrm{d}} \in \left(0, 1\right)$ such that
\begin{subequations}
\label{eq:i-ioss-lyap}
    \begin{gather}
    \alpha_{\mathrm{d, l}}(\lVert x - \tilde{x} \rVert) \leq V_\mathrm{d}(x, \tilde{x}) \leq \alpha_{\mathrm{d, u}}(\lVert x - \tilde{x} \rVert) \\
    \begin{gathered}
    V_\mathrm{d}(x^+, \tilde{x}^+) \leq \rho_{\mathrm{d}} V_\mathrm{d}(x, \tilde{x}) + 
    \sigma_{\mathrm{d}, w}(\lVert w - \tilde{w} \rVert) + \\
    \sigma_{\mathrm{d}, y}(\lVert y - \tilde{y} \rVert), \label{eq:Vo_exp_cond}
    \end{gathered}
    \end{gather}
\end{subequations}
where $x^+{=}f_w(x, u, w)$, $\tilde{x}^+ {=} f_w(\tilde{x}, u, \tilde{w})$, $y {=} h_w(x, u, w)$, $\tilde{y}{=} h_w(\tilde{x}, u, \tilde{w})$ for all $x, \tilde{x} \in \set{X}$, $u \in \set{U}$, and $w, \tilde{w} \in \set{W}$.
}

Intuitively, such  an i-IOSS Lyapunov function defines the convergence of two trajectories $x$ and $\tilde{x}$ under the same control inputs $u$ and the respective disturbances $w$ and $\tilde{w}$. In~\autoref{eq:Vo_exp_cond}, the decay rate~$\rho_{\mathrm{d}}$ defines how quickly the two trajectories converge, while $\sigma_{\mathrm{d}, w}$ and $\sigma_{\mathrm{d}, y}$ bound the effect of the disturbances and output, respectively. Given this notion of stability, we can then bound a trajectory $\{x_i\}_{i = 0}^H$ of length $H \in \set{N}$ using the knowledge of another trajectory $\{ \tilde{x}_i \}_{i = 0}^H$.
%

Incremental input-to-state stability~(i-ISS) studies the stability of trajectories in the presence of disturbances $w$~\cite{Bayer2013}. 
The notion of incremental input-to-state stability~(i-ISS) can be extended to incremental input-to-state stabilizability by considering a control Lyapunov function~(CLF). As we will see next, we can conveniently use this condition in a robust MPC framework to determine constraint tightening and thereby provide constraint satisfaction guarantees.

\definition[i-ISS CLF \cite{Kohler2021-ROFMPC}]{
\label{def:i-iss-clf}
A function $V_\mathrm{s} : \set{X} \times \set{X} \to \R_{\geq 0}$~(with subscript s for stabilizability) is called an (exponential-decrease) i-ISS CLF if there exist lower and upper bounds on $V_\mathrm{s}$ with $\alpha_{\mathrm{s, l}}, \alpha_{\mathrm{s, u}} \in \mathcal{K}_{\infty},$ bounding functions $\sigma_{\mathrm{s}, w}, \sigma_{\pi} \in \mathcal{K}$, decay rate $\rho_{\mathrm{s}} \in \left(0, 1\right)$, and a control law $\pi : \set{X} \times \set{X} \times \set{U} \to \set{U}$ such that
\begin{subequations}
    \label{eq:i-iss-clf}
    \begin{gather}
    \alpha_{\mathrm{s, l}}(\lVert x - \tilde{x} \rVert) \leq V_\mathrm{s}(x, \tilde{x}) \leq \alpha_{\mathrm{s, u}}(\lVert x - \tilde{x} \rVert) \\
    V_\mathrm{s}(x^+, \tilde{x}^+) \leq \rho_{\mathrm{s}} V_\mathrm{s}(x, \tilde{x}) +
    \sigma_{\mathrm{s}, w}(\lVert w - \tilde{w} \rVert)  \label{eq:Vs_exp_cond} \\
    \lVert \pi(\tilde{x}, x, u) - u \rVert \leq \sigma_\pi (\lVert x - \tilde{x} \rVert),
    \end{gather}
\end{subequations}
where $x^+ {=}f_w(x, u, w)$, $\tilde{x}^+ {=} f_w(\tilde{x}, \pi(\tilde{x}, x, u), \tilde{w})$ for all $x, \tilde{x} \in \set{X}$, $u \in \set{U}$, and $w, \tilde{w} \in \set{W}$.
}

The i-ISS CLF describes the convergence of two trajectories $x$ and $\tilde{x}$ under respective disturbances $w$ and $\tilde{w}$ achieved by the additional feedback $\pi$. Similar to the notion of i-ISS, the decay rate~$\rho_{\mathrm{s}}$ in~\autoref{eq:Vs_exp_cond} defines how quickly the two trajectories converge, while $\sigma_{\mathrm{s}, w}$ bounds the effect of the disturbance. 
The condition in~\autoref{eq:Vs_exp_cond} allows us to efficiently propagate the set of states that are reachable at future time steps. 

We note that there exist various methods in the literature for the synthesis of such incremental Lyapunov functions, see~\cite{Pan2022, Koelewijn2021, Kohler2020-terminal-ingredients}. 
By making a few simplifying assumptions (e.g., continuously differentiable dynamics, quadratic incremental Lyapunov function, linear $\mathcal{K}$ functions), we can rewrite the conditions for the i-IOSS Lyapunov function and i-ISS CLF as linear matrix inequalities~(LMIs).
Furthermore, as discussed in~\cite{Kohler2020-terminal-ingredients}, by using a convexifying or a gridding approach, we can turn the synthesis problem into a finite semi-definite program~(SDP), which can be solved to achieve desired decay rates~$\rho_{\mathrm{d}}$ and~$\rho_{\mathrm{s}}$.
The existence of the incremental Lyapunov functions is checked prior to closed-loop execution, and the constants and the class $\mathcal{K}$ functions are pre-computed based on either~\autoref{eq:i-ioss-lyap} or~\autoref{eq:i-iss-clf}. As we will see in~\autoref{eq:rof-mpc}, in a robust output-feedback MPC formulation, we can replace the incremental Lyapunov conditions by a set of scalar conditions using the decay rates $\rho_{\mathrm{d}}, \rho_{\mathrm{s}}$ and the class $\mathcal{K}$ functions. These scalar conditions only minimally increase the computational demand compared to a nominal MPC. 

\subsubsection*{\textbf{State Estimation Error Bounds}}
One key ingredient in a robust output-feedback MPC framework is a robustly stable observer. We next introduce Luenberger-like observers and moving horizon estimators~(MHE) with valid online estimation error bounds. We consider the Luenberger-like observer as a back-up observer if the robustness conditions for the MHE are not satisfied. 

Inspired by Luenberger observers for linear systems, a nonlinear observer is given by
\begin{equation}
\label{eq:luenberger}
\hat{x}_{k + 1} = f(\hat{x}_k, u_k) + \hat{L}(\hat{x}_k, u_k, y_k) = \hat{f}(\hat{x}_k, u_k, y_k) \,,
\end{equation}
with the estimated state $\hat{x}_k \in \R^{n_x}$ and the estimator correction $\hat{L} : \set{X} \times \set{U} \times \set{Y} \to \set{X}$ that satisfies $\hat{L}(\hat{x}, u, h(\hat{x}, u)) = 0$.

\begin{assumption}[Robustly stable observer~\cite{Kohler2021-ROFMPC}]
\label{as:robust-observer}
There exist an incremental Lyapunov function $V_{\mathrm{o}} : \set{X} \times \set{X} \to \R_{\geq 0}$~(with subscript o for observer) with lower and upper bounds $\alpha_{\mathrm{o, l}}, \alpha_{\mathrm{o, u}} \in \mathcal{K}_{\infty},$ bounding functions $ \sigma_{ \mathrm{o}, w}, \sigma_{\mathrm{o, L}}, \sigma_{\mathrm{o, L}, w} \in \mathcal{K},$ and decay rate $\rho_{\mathrm{o}} \in \left(0, 1\right)$, such that \begin{subequations}
    \begin{gather}
    \alpha_{\mathrm{o, l}}(\lVert x - \hat{x} \rVert) \leq V_\mathrm{o}(x, \hat{x}) \leq \alpha_{\mathrm{o, u}}(\lVert x - \hat{x} \rVert) \\
    V_\mathrm{o}(\hat{f}(\hat{x}, u, y), f(x, u, w)) \leq \rho_{\mathrm{o}} V_\mathrm{o}(x, \hat{x}) +
    \sigma_{\mathrm{o}, w}(\lVert w \rVert)  \\
    \lVert \hat{L}(\hat{x}, u, y) \rVert \leq \sigma_{\mathrm{o, L}} (V_\mathrm{o}(x, \hat{x})) + \sigma_{\mathrm{o, L}, w} (\lVert w \rVert), \label{eq:estimator-error}
    \end{gather}
\end{subequations}
where $y {=}h_w(x, u, w)$ for all $x, \hat{x} \in \set{X}$, $u \in \set{U}$, and $w \in \set{W}$.
\end{assumption}

To simplify the following discussion, we assume that the incremental Lyapunov function is identical to the i-IOSS Lyapunov function with $V_{\mathrm{o}} = V_{\mathrm{d}}$. Sufficient conditions for $V_{\mathrm{o}}$ being an i-IOSS Lyapunov function can be found in~\cite{Kohler2021-ROFMPC} and require affine disturbances in the state dynamics and output equation~(as assumed in~\autoref{eq:system}) and a quadratic $V_{\mathrm{o}}$. 

Given an upper bound $\bar{e}_0$ on the initial estimation error $e_0$ such that $0 \leq e_0 \leq \bar{e}_0$ and the robust stable observer described by the i-IOSS Lyapunov function $V_{\mathrm{o}}$, the error bound can be recursively updated either offline or online
\begin{align}
\raisetag{25pt}
    \begin{split}
    \label{eq:estimate-luenberger}
        \bar{e}_{k + 1, \mathrm{offline}} &= \rho_{\mathrm{o}} \bar{e}_{k} + \sigma_{\mathrm{o}, w}(\lVert \bar{w} \rVert) \,, \\
        \bar{e}_{k + 1, \mathrm{online}} &= \rho_{\mathrm{d}} \bar{e}_{k} + \sigma_{\mathrm{d}, w}(\bar{w} + \lVert \hat{w}_k \rVert) + \sigma_{ \mathrm{d}, y}(\lVert \hat{y}_k - y_k \rVert ) \,,
    \end{split}
\end{align}
where $\hat{w}_k {=} E^\intercal \left(E E^\intercal \right)^{-1}\hat{L}(\hat{x}_k, u_k, y_k) $ with $E$ defined in \autoref{eq:system} and $\hat{y}_k {=} h(\hat{x}_k, u_k)$. The online update incorporates the latest output measurement $y_k$, while the offline update requires no additional data and relies entirely on the pre-computed decay rate $\rho_{\mathrm{o}}$, $\sigma_{\mathrm{o}, w}$, and~$\bar{w}$. See~\cite{Kohler2021-ROFMPC} for a proof of the recursive updates in~\autoref{eq:estimate-luenberger}.

The (exponential-decay) i-IOSS Lyapunov function~$V_{\mathrm{o}}$ with $V_{\mathrm{o}}(\hat{x}_0, x_0) \leq \bar{e}_0$ guarantees that for all $k \in \set{I}_{\geq 0}$, $\hat{x}_k, \bar{e}_k$ from~\autoref{eq:luenberger} and~\autoref{eq:estimate-luenberger} satisfy
\begin{subequations}
    \label{eq:estimate-properties}
    \begin{align}
    V_\mathrm{o}(x_k, \hat{x}_k) &\leq \bar{e}_k \label{eq:valid_bound} \\
    \lVert \hat{L}(\hat{x}_k, u_k, y_k) \rVert &\leq \sigma_{\mathrm{o, L}} (\bar{e}_k) + \sigma_{\mathrm{o, L}, w} (\lVert \bar{w} \rVert) \,, \label{eq:nom-obsv-diff} \\
    \raisetag{25pt}
    \bar{e}_{k + i}  &\leq \rho_{\mathrm{o}}^i \bar{e}_k +
    \frac{1 - \rho_{\mathrm{o}}^i}{1 - \rho_{\mathrm{o}}} \sigma_{\mathrm{o}, w}(\lVert \bar{w} \rVert)\,, \forall i \in \set{I}_{\geq 0} \label{eq:bound_prediction} \,.
    \end{align}
\end{subequations}
See~\cite{Kohler2021-ROFMPC} for the proof. The conditions in~\autoref{eq:estimate-properties} provide desirable properties for the state estimation since they yield valid bounds on the estimation error and their prediction. 

In the following, we introduce the MHE that determines the state estimate over a backward finite-time horizon $M{>}0$. The MHE is considered as the dual of MPC for state estimation and solves a finite-horizon optimization to determine an optimized state estimate using past input and output data. 

\begin{assumption}[Continuous i-IOSS Lyapunov function] There exists a function $\sigma_{\mathrm{d}} \in \mathcal{K}$, such that for any $x, \hat{x}, \tilde{x} \in \set{X}$, the i-IOSS Lyapunov function $V_{\mathrm{d}}$ satisfies
\begin{equation}
    \lvert V_{\mathrm{d}}(\hat{x}, x) - V_{\mathrm{d}}( \tilde{x}, x) \rvert \leq \sigma_{\mathrm{d}}(\lVert \hat{x} - \tilde{x} \rVert ) \,.
\end{equation}
\end{assumption}

In the following, we use the notation $\hat{x}_{i \vert k}$ to denote the open-loop prediction from state~$\hat{x}_k$ at time step~$k + i$. 
At time step $k$, the MHE with the backward horizon $M_k = \min \{k , M\}$ and the past initial state estimate $\hat{x}_{k - M_k}$ is given by the following nonlinear program~\cite{Kohler2021-ROFMPC}:
\begin{align}
    \raisetag{70pt}
    \label{eq:mhe-opt}
    \begin{split}
        J_{M_k, \mathrm{MHE}}^*(\hat{x}_{k - M_k}) = \\ 
        \min_{\hat{w}_{ - M_k :  - 1 \vert k}, \hat{x}_{ - M_k \vert k}} &
        \sum_{j = 1}^{M_k} \rho_{\mathrm{d}}^{j - 1} (\sigma_{\mathrm{d}, w}(\bar{w} + \lVert \hat{w}_{ - j \vert k} \rVert) + \\ & \quad \quad \sigma_{\mathrm{d}, y}(\lVert \hat{y}_{ - j \vert k} - y_{k -j} \rVert )) + \\
        & \quad \quad \rho_{\mathrm{d}}^{M_k} \bar{e}_{k - M_k} + \\ & \quad \quad \rho_{\mathrm{d}}^{M_k}  \sigma_{\mathrm{d}} (\lVert \hat{x}_{ - M_k \vert k} - \hat{x}_{k - M_k} \rVert) \\
        \text{s.t.} & \quad \hat{x}_{ - j + 1 \vert k} = f_w(\hat{x}_{ - j \vert k}, u_{k - j}, \hat{w}_{- j \vert k})\,, \\
        & \quad \hat{y}_{- j \vert k} = h_w(\hat{x}_{ - j \vert k}, u_{k - j}, \hat{w}_{ -j \vert k}) \,, \\
        & \quad \forall j \in \set{I}_{\left[1, M_k\right]} \,.
    \end{split}
\end{align}
The resulting state estimate and estimation error bound can be determined by propagating the minimizers of~\autoref{eq:mhe-opt} $\hat{w}^*_{ - M_k :  - 1 \vert k}, \hat{x}^*_{ - M_k \vert k}$ until $\hat{x}_{0 \vert k}^*$ is reached:
\begin{align}
    \begin{split}
    \label{eq:estimate-mhe}
        \hat{x}_{k, \mathrm{MHE}} &= \hat{x}_{0 \vert k}^* \,, \\
        \hat{e}_{k, \mathrm{MHE}} &= J^*_{M_k, \mathrm{MHE}}(\hat{x}_{k - M_K}) \,. 
    \end{split}
\end{align}

As the MHE requires additional assumptions to be robustly stable, we can instead check conditions~\autoref{eq:nom-obsv-diff} and~\autoref{eq:bound_prediction} on the estimation error for the MHE update to guarantee robustness of the MHE update. In the following, we assume that the resulting estimation error bound from the MHE is valid according to~\autoref{eq:nom-obsv-diff} and~\autoref{eq:bound_prediction}. In practice, these conditions must be checked at every step. If the conditions are not met, the backup Luenberger-like observer is used. 


\subsubsection*{\textbf{Prediction Error Bounds}}
In this part, we highlight the use of incremental Lyapunov functions for  tube-based predictions of the combined error bounds originating from the estimation error and the additive disturbance on the dynamics. This requires the analysis of the incremental stability between the true state $x$, the state estimate $\hat{x}$, and a nominal state $\bar{x}$. 

Under the assumption that $V_{\mathrm{s}}$ is an i-ISS CLF, there also exist functions $\sigma_{\mathrm{s, o}}, \sigma_{\mathrm{s, o}, w} \in \mathcal{K}$ such that:
\begin{equation}
    V_{\mathrm{s}}(\bar{x}^+, \hat{x}^+) \leq \rho_{\mathrm{s}} V_{\mathrm{s}}(\bar{x},\hat{x}) + \sigma_{\mathrm{s, o}} (V_{\mathrm{o}}(\hat{x}, x) ) +  \sigma_{\mathrm{s, o}, w} (\bar{w}) \,,
\end{equation}
where $u = \pi(\hat{x}, \bar{x}, \bar{u})$, $y = h_w(x, u, w)$,  $\bar{x}^+ = f(\bar{x}, \bar{u})$, and $\hat{x}^+=\hat{f}(\hat{x}, u, y)$  (see~\cite{Kohler2021-ROFMPC} for a proof of this implication). 
The conditions on $V_{\mathrm{s}}$ can be used to predict the combined errors from the estimation and the disturbed state dynamics. 

\subsubsection*{\textbf{Tube-based Robust Output-feedback MPC}}
The robust output-feedback MPC minimizes a user-defined continuous stage cost $\ell : \set{X} \times \set{U} \times \R_{\geq 0} \times \R_{\geq 0} \to \R$, which allows for additional arguments to also include the bounds on the estimation and prediction errors. Robust constraint satisfaction is guaranteed via the proper design of a continuous terminal cost $V_{\mathrm{f}} : \set{X} \times \R_{\geq 0} \times \R_{\geq 0} \to \R$ and a terminal constraint set $\set{X}_{\mathrm{f}} \subseteq \set{X} \times \R_{\geq 0} \times \R_{\geq 0}$. 
At time step $k$, for a given state estimate from the Luenberger-like observer~\autoref{eq:estimate-luenberger} and~\autoref{eq:luenberger} or the MHE~\autoref{eq:estimate-mhe}, the robust output-feedback MPC open-loop optimization problem over a finite horizon $N\in\set{I}_{\geq 1}$ is:
\begin{subequations}
    \raisetag{110pt}
    \label{eq:rof-mpc}
    \begin{align}
        J_{N, \mathrm{MPC}}^* & (\hat{x}_k, \bar{e}_k) = \\
        \min_{\bar{u}_{0 : N - 1 \vert k}, \bar{x}_{0 \vert k }} & \sum_{i = 0}^{N - 1} l(\bar{x}_{i \vert k}, \bar{u}_{i \vert k}, \bar{e}_{i \vert k}, \bar{s}_{i \vert k}) + \\
        & \quad\quad\quad V_{\mathrm{f}}(\bar{x}_{N \vert k}, \bar{e}_{N \vert k}, \bar{s}_{N \vert k}) \label{eq:init-tube}\\
        \text{s.t.} & \quad \bar{s}_{0 \vert k} = V_{\mathrm{s}} (\bar{x}_{0 \vert k}, \hat{x}_k ) \,, \\
        & \quad \bar{x}_{i + 1 \vert k} = f(\bar{x}_{i \vert k}, \bar{u}_{i \vert k}) \,, \forall i \in \set{I}_{\left[0, N- 1 \right]} \,, \label{eq:nom-dynamics} \\
        & \quad \bar{e}_{j \vert k} = \frac{1 - \rho_{\mathrm{d}}^j}{1 - \rho_{\mathrm{d}}} \sigma_{\mathrm{o}, w}(\bar{w}) + \rho_{\mathrm{d}}^j \bar{e}_k \,, \forall j \in \set{I}_{\left[0, N \right]}  \,, \label{eq:estimation-dyn-mpc} \\
        & \quad \bar{s}_{i + 1 \vert k} = \rho_{\mathrm{s}} \bar{s}_{i \vert k} + \sigma_{\mathrm{s, o}} (\bar{e}_{i \vert k}) + \sigma_{\mathrm{s, o}, w} (\bar{w}) \,, \label{eq:tube-dyn-mpc} \\
        & \quad (x_{i \vert k}, \pi(\hat{x}_{i \vert k}, \bar{x}_{i \vert k}, \bar{u}_{i \vert k})) \in \set{Z} \,, \label{eq:consrtaints}\\
        & \quad  V_{\mathrm{s}}(\bar{x}_{i \vert k}, \hat{x}_{i \vert k}) \leq \bar{s}_{i \vert k} \,, \forall x_{i \vert k}, \hat{x}_{i \vert k} \,, \label{eq:estimate-in-tube}\\
        & \quad V_{\mathrm{o}}(\hat{x}_{i \vert k}, x_{i \vert k}) \leq \bar{e}_{i \vert k} \,,  \forall x_{i \vert k}, \hat{x}_{i \vert k} \,, \label{eq:state-in-tube}\\
        & \quad (\bar{x}_{N \vert k}, \bar{e}_{N \vert k}, \bar{s}_{N \vert k}) \in \set{X}_{\mathrm{f}} \label{eq:terminal-constraint} \,.
    \end{align}
\end{subequations}
\autoref{eq:init-tube} guarantees that the initial nominal state is inside the set around the state estimate given by $V_{\mathrm{s}}$. The constraints in~\autoref{eq:nom-dynamics}, \autoref{eq:estimation-dyn-mpc}, and~\autoref{eq:tube-dyn-mpc} specify the nominal dynamics, estimation error bound and prediction error bound propagation, respectively. \autoref{eq:consrtaints} guarantees constraint satisfaction for the true state and control input. Proper containment inside the tubes given by $V_{\mathrm{s}}$ and $V_{\mathrm{o}}$ for the nominal and true state is achieved by~\autoref{eq:estimate-in-tube} and~\autoref{eq:state-in-tube}. Finally, the terminal constraint is satisfied with~\autoref{eq:terminal-constraint}.

The minimizers of~\autoref{eq:rof-mpc} are $\bar{u}^*_{0 : N - 1 \vert k}, \bar{x}*_{0 \vert k }$.
At each time step, only the first optimal control input $\bar{u}_{0 \vert k}^*$ is applied to the system through $u_k = \pi(\hat{x}_k, \bar{x}_{0 \vert k}^*, \bar{u}_{0 \vert k}^*)$. Then the open-loop optimization problem is solved again at the next time step using an updated state estimate $\hat{x}_{k + 1}$ and error bound $\bar{e}_{k + 1}$ over a shifted time horizon.  

Properly designed terminal ingredients provide constraint satisfaction and recursive feasibility of the robust output-feedback MPC. 
\begin{assumption}[Terminal constraint set~\cite{Kohler2021-ROFMPC}]
\label{as:terminal-set}
There exists a control law $\pi_{\mathrm{f}} : \set{X} \to \set{U}$ such that for all $(\bar{x}, \bar{e}, \bar{s}) \in \set{X}_{\mathrm{f}}$ and $x, \hat{x} \in \set{X}$ satisfying $V_{\mathrm{s}}(\bar{x}, \hat{x}) \leq \bar{s}$ and $V_{\mathrm{o}}(\hat{x}, x) \leq \bar{e}$, and for all $\bar{s}^+, \bar{e}^+ \in \R_{\geq 0}$ satisfying $\bar{s}^+ \leq \rho_{\mathrm{s}} \bar{s} + \sigma_{\mathrm{s, o}}(\bar{e}) + \sigma_{\mathrm{s, o}, w}(\bar{w})$ and $\bar{e}^+ \leq \rho_{\mathrm{o}} \bar{e} + \sigma_{\mathrm{o}, w} (\bar{w})$, it holds that
\begin{subequations}
\begin{align}
        (\bar{x}^+, \bar{e}^+, \bar{s}^+) \in \set{X}_{\mathrm{f}} \,, \quad
    (x, \pi(\hat{x}, \bar{x}, \bar{u})) \in \set{Z} \,,
\end{align}
\end{subequations}
where $\bar{x}^+{=}f(\bar{x}, \bar{u})$ and $\bar{u}{=}\pi_{\mathrm{f}}(\bar{x})$. 
\end{assumption}

Recursive feasibility and boundedness of the closed-loop cost can be shown using a terminal constraint set satisfying \autoref{as:terminal-set} and some  additional assumptions on the stage cost~$\ell$ and the terminal cost function~$V_{\mathrm{f}}$~\cite{Kohler2021-ROFMPC}. 

Under the assumption of continuous constraint functions and the initial nominal state being equal to the initial state estimate, $\bar{x}^*_{0 \vert k} = \hat{x}_{k}$ (this implies $\bar{s}^*_{0 \vert k} = 0$), the constraint tightening can be pre-computed. The resulting computational demand is comparable to a nominal MPC~\cite{Kohler2021-ROFMPC}. 

We use the robust output-feedback MPC framework as the basis for our proposed safety filter to provide safety guarantees for constrained nonlinear systems with bounded estimation errors. 

\subsection{Model Predictive Safety Certification}

The model predictive safety certification~(MPSC) method is inspired by model predictive control techniques. Instead of providing an optimal control input, the MPSC takes an arbitrary control input and certifies the provided input either as safe or unsafe~\cite{Wabersich2021}. If the control input is considered as safe, it passes through the MPSC without modification. Otherwise, the control input is deemed unsafe and is minimally modified (according to a distance measure) such that the control input is safe. This relies on techniques from MPC to guarantee constraint satisfaction for all time steps. Guarantees are provided through the existence of a safe set~$\set{S}_{\mathrm{f}}$: 
\begin{assumption}[Safe set]
\label{as:safe-set}
There exists a control policy $\pi_{\mathrm{safe}} : \set{X} \to \set{U}$ such that for all $(\bar{x}_k, \bar{e}_k, \bar{s}_k) \in \set{S}_{\mathrm{f}}$ and for all $j \in \set{I}_{\geq 0} $ and all $x_{k + j}, \hat{x}_{k + j} \in \set{X}$ satisfying $V_{\mathrm{s}}(\bar{x}_{k + j}, \hat{x}_{k + j}) \leq \bar{s}_{k + j}$ and $V_{\mathrm{o}}(\hat{x}_{k + j}, x_{k + j}) \leq \bar{e}_{k + j}$, it holds that
\begin{equation}
    (x_{k + j}, \pi(\hat{x}_{k + j}, \bar{x}_{k + j}, \bar{u}_{k + j})) \in \set{Z} 
\end{equation}
with $\bar{x}_{k + j + 1}{=}f(\bar{x}_{k + j}$, $\bar{u}_{k + j})$, $\bar{u}_{k + j}{=}\pi_{\mathrm{safe}}(\bar{x}_{k + j})$ and with any $\bar{s}_{k + j + 1}, \bar{e}_{k + j + 1} \in \R_{\geq 0}$ satisfying $\bar{s}_{k + j + 1} \leq \rho_{\mathrm{s}} \bar{s}_{k + j} + \sigma_{\mathrm{s, o}}(\bar{e}_{k + j}) + \sigma_{\mathrm{s, o}, w}(\bar{w})$ and $\bar{e}_{k + j + 1} \leq \rho_{\mathrm{o}} \bar{e}_{k + j} + \sigma_{\mathrm{o}, w} (\bar{w})$. 
\end{assumption}

The assumption of the existence of safe set $\set{S}_{\mathrm{f}}$ is less restrictive than the assumption of a robustly positive invariant terminal set $\set{X}_{\mathrm{f}}$ as in~\autoref{as:terminal-set}. However, guaranteeing constraint satisfaction for all future time steps requires additional attention (cf. proof for~\autoref{th:rof-psf}). 


We present the idea of MPSC using nominal dynamics and full-state measurements, such that $x_k = \bar{x}_k$ and $u_k = \pi(\hat{x}_{k}, \bar{x}_{k}, \bar{u}_{k}) = \pi(\bar{x}_{k}, \bar{x}_{k}, \bar{u}_{k})= \bar{u}_{k}$. The safety filter framework can be extended to uncertain systems using a formulation as in robust MPC. 
The main difference between the MPSC and a nominal MPC is the objective function. For the MPSC, the objective is the squared error between the first optimal control input $u^*_{0 \vert k}$ and the provided control input at the current state $\tilde{u} = \pi_{\mathcal{L}}(x_{k})$, where $\pi_{\mathcal{L}}: \set{X} \times \set{U}$ is an arbitrary and potentially unsafe control policy. Then the open-loop optimization problem for the MPSC at time step~$k$ under the assumption of zero disturbances and perfect state measurements is given by: 
\begin{align}
\label{eq:mpsc-opt}
    \begin{split}
        J_{N, \mathrm{MPSC}}^*(x_k) = \min_{u_{0 : N - 1 \vert k}} & \lVert \pi_\mathcal{L}(x_{k}) - u_{0 \vert k} \rVert^2_2 \\
        \text{s.t.} & \quad \forall i \in \set{I}_{\left[0, N- 1 \right]}  \,, \\
        & \quad x_{i + 1 \vert k} = f(x_{i \vert k}, u_{i \vert k}) \,, \\
        & \quad (x_{i \vert k}, u_{i \vert k}) \in \set{Z} \,, \\
        & \quad x_{0 \vert k} = x_k \,, \\
        & \quad x_{N \vert k} \in \set{S}_{\mathrm{f}} \,.
    \end{split}
\end{align}
The certified control inputs that can be safely applied to the system are given by the minimizer $u_k = u_{0 \vert k}^*$ in an MPC fashion.
Proofs for constraint satisfaction for all future time steps using an adaptive horizon approach and guarantees for handling uncertain state dynamics can be found in~\cite{Wabersich2021}.

Previous results assume full-state measurements without noise. In this work, we extend the model predictive safety certification optimization problem to handle uncertain output measurements using state estimates with valid error bounds.  

\section{\NAMECAPS~(\acronym)}
In this section, we present our proposed \acronym~(see \autoref{fig:blockdiagram}). The \acronym~takes an input from arbitrary state-feedback control policies and either certifies the input as safe or modifies the desired control input minimally (according to the $\ell_2$-norm) to still guarantee safety for uncertain constrained nonlinear systems. 


\subsection{Overview of the \acronym~Algorithm}
\label{subsec:algorithm_overview}
As we are dealing with noisy and partial state measurements, we leverage the state estimation techniques from the previous section to determine state estimates with valid error bounds. 
To simplify the formulation, we assume that the MHE is a robustly stable observer. In practice, this requires checking the conditions~\autoref{eq:nom-obsv-diff} and~\autoref{eq:bound_prediction}. The estimation error bound at time step $k \in \set{I}_{\geq 1} $ is computed as
\begin{equation}
    \label{eq:estimate}
    \bar{e}_{k} = \min \{ \bar{e}_{k, \mathrm{offline}}, \bar{e}_{k, \mathrm{online}}, \bar{e}_{k, \mathrm{MHE}} \} \,,
\end{equation}
to determine the smallest upper bound and the state estimate is then either given by the associated MHE or the Luenberger-like observer with
\begin{equation}
    \label{eq:estimate-bound}
    \hat{x}_k = \begin{cases} 
    \hat{x}_{k, \mathrm{MHE}} & \mathrm{if}~\bar{e}_{k} = \bar{e}_{k, \mathrm{MHE}} \,, \\
    \hat{f}(\hat{x}_{k - 1}, u_{k - 1}, y_{k - 1}) & \mathrm{otherwise} \,.
    \end{cases}
\end{equation}
The pair of state estimate~$\hat{x}_k$ and estimation error bound~$\bar{e}_k$ are then used to initialize our predictive safety filter. 

By combining the robust output-feedback MPC and the MPSC, we propose the following open-loop optimization problem for the robust output-feedback predictive safety filter:
\begin{subequations}
\label{eq:rof-psf}
    \begin{align}
        J_{N, \mathrm{\acronym}}^* & (\hat{x}_k, \bar{e}_k) = \nonumber \\
        \min_{\bar{u}_{0 : N - 1 \vert k}, \bar{x}_{0 \vert k }} & \lVert \pi_\mathcal{L}(\hat{x}_{k}) - \pi(\hat{x}_{k}, \bar{x}_{0 \vert k}, \bar{u}_{0 \vert k}) \rVert^2_2 \label{eq:safe-objective}\\
        \text{s.t.} & \quad \bar{s}_{0 \vert k} = V_{\mathrm{s}} (\bar{x}_{0 \vert k}, \hat{x}_k ) \,, \\
        & \quad \bar{x}_{i + 1 \vert k} = f(\bar{x}_{i \vert k}, \bar{u}_{i \vert k}) \,, \forall i \in \set{I}_{\left[0, N- 1 \right]} \,, \label{eq:state-dyn} \\
        & \quad \bar{e}_{j \vert k} = \frac{1 - \rho_{\mathrm{d}}^j}{1 - \rho_{\mathrm{d}}} \sigma_{\mathrm{o}, w}(\bar{w}) + \rho_{\mathrm{d}}^j \bar{e}_k \,, \forall j \in \set{I}_{\left[0, N \right]}  \,, \label{eq:est-dyn} \\
        & \quad \bar{s}_{i + 1 \vert k} = \rho_{\mathrm{s}} \bar{s}_{i \vert k} + \sigma_{\mathrm{s, o}} (\bar{e}_{i \vert k}) + \sigma_{\mathrm{s, o}, w} (\bar{w}) \,, \label{eq:tube-dyn} \\
        & \quad (x_{i \vert k}, \pi(\hat{x}_{i \vert k}, \bar{x}_{i \vert k}, \bar{u}_{i \vert k})) \in \set{Z} \,, \\
        & \quad  V_{\mathrm{s}}(\bar{x}_{i \vert k}, \hat{x}_{i \vert k}) \leq \bar{s}_{i \vert k} \,, \forall x_{i \vert k}, \hat{x}_{i \vert k} \,, \\
        & \quad V_{\mathrm{o}}(\hat{x}_{i \vert k}, x_{i \vert k}) \leq \bar{e}_{i \vert k} \,,  \forall x_{i \vert k}, \hat{x}_{i \vert k} \,, \\
        & \quad (\bar{x}_{N \vert k}, \bar{e}_{N \vert k}, \bar{s}_{N \vert k}) \in \set{S}_{\mathrm{f}} \label{eq:safe-terminal-constraint}\,.
    \end{align}
\end{subequations}
In the proposed~\acronym~we use the objective function from~\autoref{eq:mpsc-opt} in~\autoref{eq:safe-objective} to certify arbitrary control policies and substitute the terminal set $\set{X}_{\mathrm{f}}$ for the more general safe set $\set{S}_{\mathrm{f}}$ in~\autoref{eq:safe-terminal-constraint}.

The minimizers of~\autoref{eq:rof-psf} are denoted by $\bar{u}^*_{0 : N - 1 \vert k}, \bar{x}*_{0 \vert k }$.
We apply $u_k = \pi(\hat{x}_{k}, \bar{x}^*_{0 \vert k}, \bar{u}^*_{0 \vert k})$ to the system and solve the open-loop optimization again at the next time step with a receding horizon. If the problem is infeasible at the next time step, which is possible due to the more general assumption on the safe set~$\set{S}_{\mathrm{f}}$, the procedure in~\autoref{alg:rof-psf} must be followed to guarantee constraint satisfaction for all future time steps. If $\pi_\mathcal{L}(\hat{x}_{k})$ is a feasible control input then $u_k = \pi_\mathcal{L}(\hat{x}_{k})$. Otherwise $u_k$ is the minimal modification of $\pi_\mathcal{L}(\hat{x}_{k})$ such that the optimization in~\autoref{eq:rof-psf} is feasible.


The proposed algorithm is given in~\autoref{alg:rof-psf}, where the additional subscript $\tilde{N}$ in lines 9 and 10 of~\autoref{alg:rof-psf} indicates that this nominal state or control input is the result of an optimization with reduced horizon~$\tilde{N} < N$. 
\begin{algorithm}
\caption{Robust predictive output-feedback safety filter}\label{alg:rof-psf}
\begin{algorithmic}[1]
\While{True}
    \State Update $\hat{x}_k$ and $\bar{e}_k$ using~\autoref{eq:estimate} and~\autoref{eq:estimate-bound}, respectively.
    \If{\ref{eq:rof-psf} is feasible for horizon $N$} 
        \State $k_{\mathrm{feasible}} \gets k$ \Comment{Update feasible time step}
        \State $u_k \gets \pi(\hat{x}_{k}, \bar{x}^*_{0 \vert k}, \bar{u}^*_{0 \vert k})$
    \Else
        \If{$k < N + k_{\mathrm{feasible}}$}
            \State Solve \ref{eq:rof-psf} for horizon  $\tilde{N} := N - (k - k_{\mathrm{feasible}})$
            \State $u_k \gets \pi(\hat{x}_{k}, \bar{x}^*_{0 \vert k, \tilde{N}}, \bar{u}^*_{0 \vert k, \tilde{N}})$
            \State $\bar{x}_k \gets \bar{x}^*_{0 \vert k, \tilde{N}}$
        \Else 
            \State $u_k \gets \pi(\hat{x}_{k}, \bar{x}_{k}, \pi_\mathrm{safe}(\bar{x}_k))$
            \State $\bar{x}_k \gets f(\bar{x}_k, \pi_\mathrm{safe}(\bar{x}_k))$ 
        \EndIf
    \EndIf 
    \State Apply the certified control input $u_k$
    \State $k \gets k + 1$
\EndWhile
\end{algorithmic}
\end{algorithm}


The formulation can be further simplified as in the robust output-feedback MPC with pre-computed constraint tightening and tube size 0 at the initial nominal state~\cite{Kohler-simple-ROF}. 



\subsection{\acronym~Constraint Satisfaction Guarantees}
\label{subsec:constraint_satisfaction}
By leveraging the results in~\cite{Kohler2021-ROFMPC, Wabersich2021}, we show that the proposed \acronym~guarantees constraint satisfaction for all time steps.
\begin{theorem}
\label{th:rof-psf}
Let Assumptions~\ref{as:robust-observer} and~\ref{as:safe-set} hold. Suppose that the system in~\autoref{eq:system} admits an (exponential-decay) i-IOSS Lyapunov function~(\autoref{def:i-ioss-lyap})  and an (exponential decrease) i-ISS CLF~(\autoref{def:i-iss-clf}). If the optimization in~\autoref{eq:rof-psf} is feasible at $k = 0$ with a valid initial estimation error bound $V_{\mathrm{o}}(\hat{x}_0, x_0) \leq \bar{e}_0$, then the~\acronym~described in~\autoref{alg:rof-psf} guarantees constraint satisfaction for all $k \in \set{I}_{\geq 0}$. Furthermore, if the safe set $\set{S}_{\mathrm{f}}$ also satisfies~\autoref{as:terminal-set}, then the optimization in~\autoref{eq:rof-psf} is recursively feasible. 
\end{theorem}

\begin{proof}
Suppose that~\autoref{eq:rof-psf} is feasible at time step $k \in \set{I}_{\geq 0}$. The sequence of nominal optimal control inputs is $\{ \bar{u}_{0 \vert k}^*, \dots, \bar{u}_{ N - 1 \vert k}^* \}$ and the initial optimal nominal state is $\bar{x}^*_{0 \vert k}$. The true state and input satisfy $(x_k, u_k) \in \set{Z}$ with $u_k = \pi(\hat{x}_{k}, \bar{x}^*_{0 \vert k}, \bar{u}^*_{0 \vert k})$ due to the over-approximation of the error bounds $\bar{e}_k$ and $\bar{s}^*_{0 \vert k}$ guaranteed by the i-IOSS Lyapunov function and the i-ISS CLF under~\autoref{as:robust-observer}. Then we can safely apply the control input $u_k$. 

If the optimization in~\autoref{eq:rof-psf} is feasible at time $(k + 1)$, then constraint satisfaction for the true state $(x_{k + 1}, u_{k + 1}) \in \set{Z}$ is again guaranteed by the over-approximation of the error bounds and we apply $u_{k + 1}$. 

If the optimization is infeasible at time $(k{+}1)$, we can construct a feasible solution using a previous feasible solution for a reduced horizon $(N{-}1)$. This is guaranteed to be feasible since we  already know that there exists at least one feasible nominal control input sequence with $\{ \bar{u}_{1 \vert k}^*, \dots, \bar{u}_{N - 1 \vert k}^* \}$ and the initial nominal state $\bar{x}_{0 \vert k + 1} = \bar{x}^*_{1 \vert k}$. Feasibility of the optimization in~\autoref{eq:rof-psf} with horizon $(N{-}1)$ again guarantees $(x_{k + 1}, u_{k + 1}) \in \set{Z}$ with $u_{k + 1} = \pi(\hat{x}_{k + 1}, \bar{x}^*_{0 \vert k + 1, N - 1}, \bar{u}^*_{0 \vert k + 1, N- 1})$ and we apply~$u_{k + 1}$. 

If the optimization with the full horizon $N$ is consecutively infeasible for the next $N - 1$ steps, then feasibility of the optimization problem with horizon of length $1$ at the previous time step guarantees that $(\bar{x}_{1 \vert k + N - 1, 1}, \bar{e}_{1 \vert k + N - 1, 1}, \bar{s}_{1 \vert k + N - 1, 1}) \in \set{S}_{\mathrm{safe}}$. Then by~\autoref{as:safe-set}, all future states and inputs guarantee $(x_{k + N + i}, u_{k + N + i}) \in \set{Z}$ for all $i \in \set{I}_{\geq 0}$. 

In case the safe set~$\set{S}_{\mathrm{f}}$ satisfies~\autoref{as:terminal-set}, we recover recursive feasibility and constraint satisfaction and require no adaptation of the horizon, see~\cite{Kohler2021-ROFMPC}. 
\end{proof}

The proof above shows that, despite the uncertainty in dynamics and output equations, the proposed \acronym~is still able to achieve constraint satisfaction at every time step. This result enables the certification of control inputs from arbitrary control policies. 


\section{NUMERICAL EXAMPLE}
In this section, we demonstrate that~\acronym~achieves constraint satisfaction for uncertain nonlinear systems despite an arbitrary control policy. 
We apply the proposed \acronym~on the following simulated nonlinear mass-spring-damper system from~\cite{Magni2003} with an output measurement of the first element of the state, similar to~\cite{Kohler-simple-ROF}:
\begin{align}
    \begin{split}
        &\dot{x}_1 = x_2 \,,\: \dot{x}_2 = \frac{1}{M} \big ( - k_0 \exp(- x_1) \: x_1 - h_d x_2 + u \big ) \,, \\
        &y = x_1\,,
    \end{split}
\end{align}
where $M = 1$, $k_0 = 0.33$, and $h_d = 1.1$. The discrete-time model is determined based on a fourth-order Runge-Kutta method with sampling time $\Delta t = 0.25~\mathrm{s}$, which yields the discrete-time dynamics $f(x, u)$. We consider additive disturbances with $E w = \begin{bmatrix} \Delta t \cdot w_1 & \frac{\Delta t}{M} w_2 \end{bmatrix}^\intercal$ and $Fw = w_3\,,$
where $w \in \set{W} = \{ w \in \R^3 : \lVert w \rVert_\infty \leq 0.01 \}$. We enforce the following constraints point-wise in time $\set{Z} = [-0.85, 0.85] \times [-2, 2] \times [-6, 6]$. 


We derive LMIs for the incremental Lyapunov functions assuming quadratic incremental Lyapunov functions, linear $\mathcal{K}$ and $\mathcal{K}_{\infty}$ functions, and linear feedback $\pi$, similar to~\cite{Kohler2020-terminal-ingredients}. Using a gridding approach, we can solve the resulting SDPs over the entire constraint set for the discrete-time system since we assume global properties. We run this computation offline in Matlab using YALMIP~\cite{Lofberg2004} and the solvers MOSEK~\cite{mosek} and SDPT3~\cite{sdpt3}.
This yields $\rho_{\mathrm{d}} = 0.74$,  $\rho_{\mathrm{o}} = 0.67$,  $\rho_{\mathrm{s}} = 0.78$ and $\sigma_{\mathrm{o}, w}(a) = 2.25 a$,  $\sigma_{ \mathrm{s, o}}(a) = 1.04 a$, and $ \sigma_{ \mathrm{s}, w}(a) = 2.23 a$.

For simplicity, we design a safe set~$\set{S}_{\mathrm{f}}$ defined by a quadratic Lyapunov function $V_{\mathrm{f}}$, that also satisfies~\autoref{as:terminal-set}. The safe set and the associated controller are determined as in~\cite{Rawlings2017} and the size of the safe set is determined online based on the error bounds as in~\cite{Kohler-simple-ROF}. Since the constraints are polytopic, we use the constraint tightening strategy from~\cite{Kohler2021-ROFMPC}.  

The \acronym~uses a horizon $N = 40$ and the MHE uses a backward horizon $M = 10$. The uncertified control inputs at every time step $k$ are obtained from an arbitrary sinusoidal control input signal.
We assume $\hat{x}_0 = x_0 \iff \bar{e}_0 = 0$ with the initial state $x_0 = \begin{bmatrix} 0.79 & 0.7 \end{bmatrix}^\intercal$. 
We implement the online computation~(see~\autoref{alg:rof-psf}) in Matlab using Casadi~\cite{Andersson2019} and solve the MPC and MHE using IPOPT~\cite{ipopt2006}, with the MHE optimization limited to $1\mathrm{e}3$ iterations. 

The closed-loop behavior of the proposed \acronym~is shown in~\autoref{fig:closed-loop}. This highlights the successful certification of an arbitrary control policy, which would otherwise lead to constraint violations. The proposed \acronym~is able to modify the control inputs at every time step to achieve constraint satisfaction despite state disturbances, measurement uncertainties, and estimation errors. We emphasize that the safety filter allows safe operation even in proximity to a constraint boundary due to the over-approximated error bounds. 

\begin{figure*}[tb]
    \centering
    \subfloat{{\includegraphics[width=0.95\columnwidth]{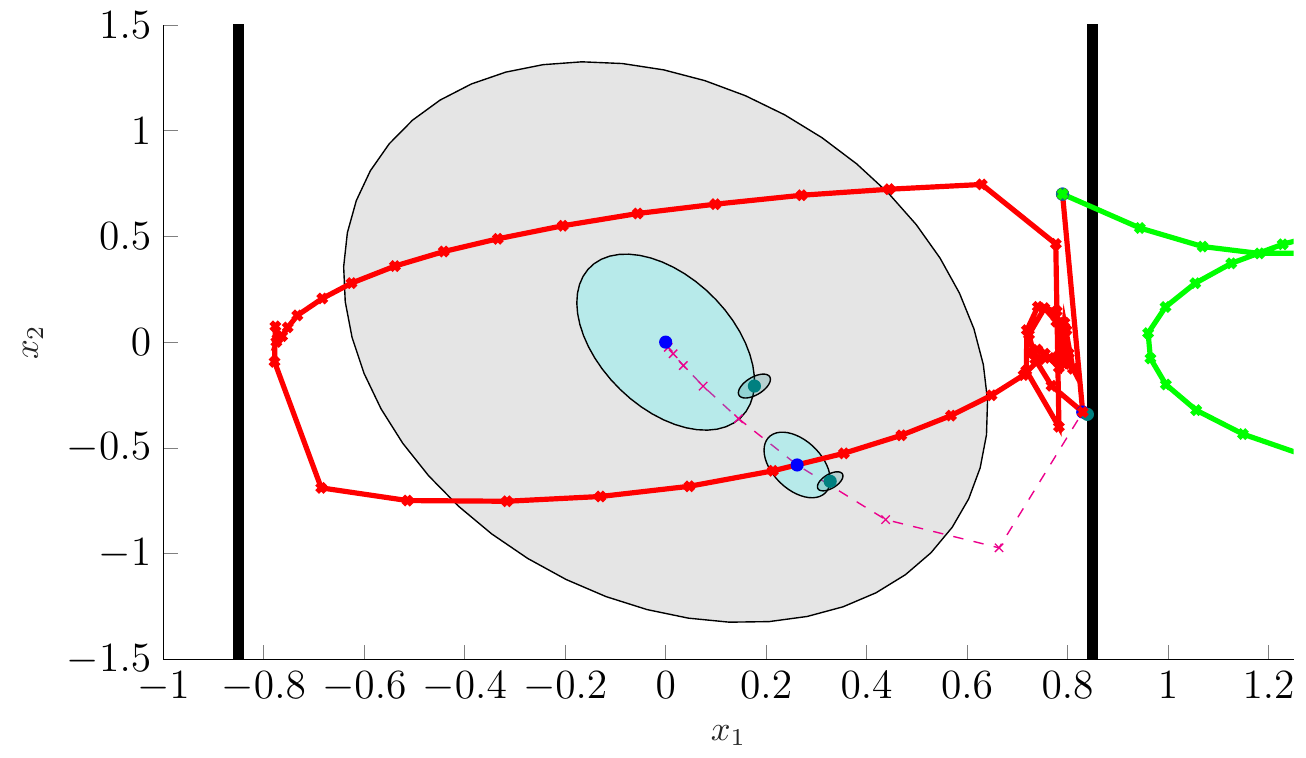} }}%
    \qquad
    \subfloat{{\includegraphics[width=0.95\columnwidth]{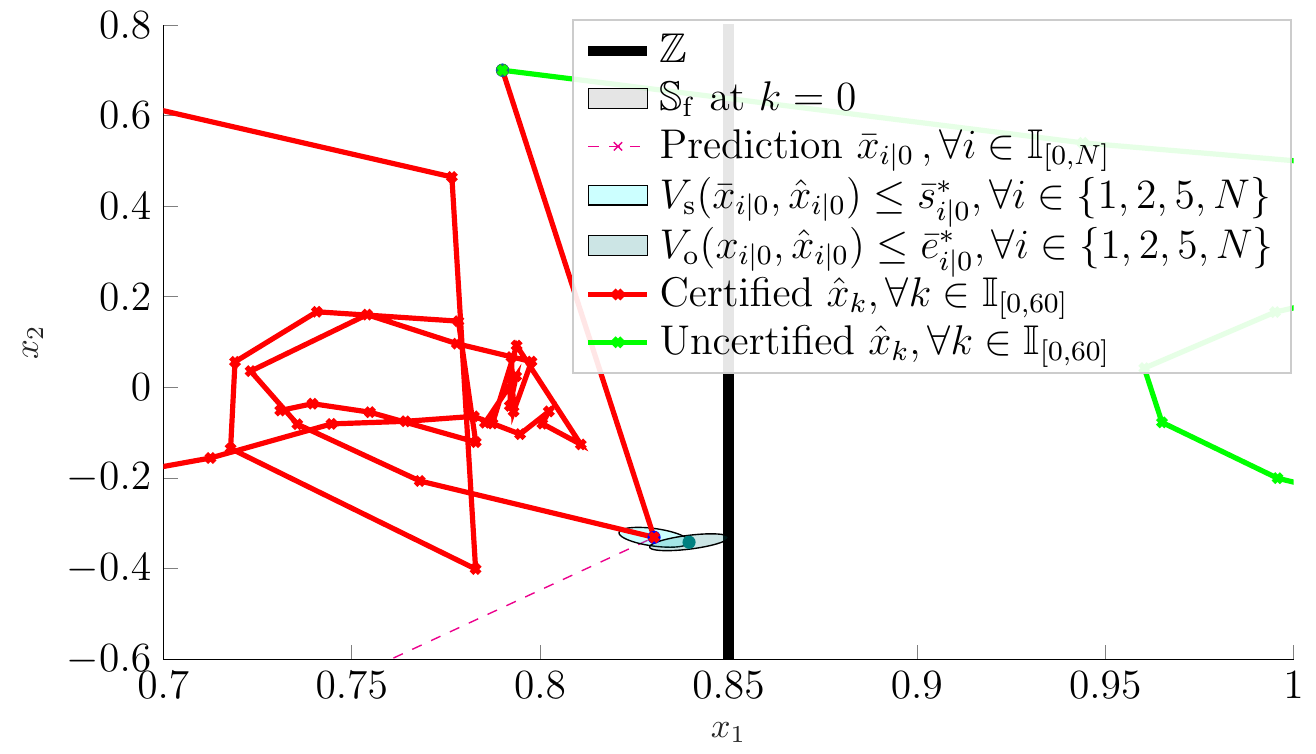} }}%
    \caption{Closed-loop behavior of the proposed \acronym~for 60 time steps. The constraint boundaries (black) and closed-loop trajectories of the estimated state for \textit{(i)} the unfiltered control input $u_k = \pi_{\mathcal{L}}(k)$ (green) and \textit{(ii)}  the certified control input as provided through the \acronym~(red) are shown. The first open-loop trajectory (magenta, dashed) starting from $x_0 = \hat{x}_0 = \begin{bmatrix} 0.79 & 0.7 \end{bmatrix}^\intercal$ and the predicted ellipsoids for the estimation error bound given by $V_{\mathrm{o}}$ (teal) and the prediction error bound given by $V_{\mathrm{s}}$ (blue) for the predicted time steps $k = \{1, 2, 5, N = 40\} $ are also displayed. The terminal set (gray) is shown for the first open-loop prediction. 
    \textbf{Left:} The figure on the left shows that the unfiltered control input immediately leads to constraint violations. In contrast, the certified control inputs from our robust predictive output-feedback safety filter achieve constraint satisfaction for all time steps. 
    \textbf{Right:} The figure on the right shows a close-up of (a). This highlights constraint satisfaction of the \acronym. Furthermore, the worst-case ellipsoid for $V_{\mathrm{o}}$ (teal) touches the constraint, which validates the constraint tightening. 
    }%
    \label{fig:closed-loop}%
\end{figure*}

\section{CONCLUSIONS}

In this paper, we proposed a robust predictive output-feedback  safety filter (\acronym) for certifying arbitrary control inputs applied to disturbed nonlinear systems without full-state measurements. 
The efficacy of the proposed \acronym~approach for guaranteeing constraint satisfaction under uncertainties is proved in theory and demonstrated using a mass-spring-damper system as a numerical example. 


As future work,  we plan to generalize the proposed \acronym~approach to a broader class of uncertain systems by incorporating probabilistic learning techniques and use the proposed approach to guide reinforcement learning to improve sampling efficiency.










\bibliographystyle{IEEEtran}
\bibliography{IEEEabrv,references}

\addtolength{\textheight}{-3cm}   

\end{document}